\documentclass[conference]{IEEEtran}
\usepackage{graphicx}
\usepackage{amsthm}
\usepackage{epsfig}
\usepackage{latexsym}
\usepackage{amsfonts}
\usepackage{here}
\usepackage{rawfonts}
\usepackage[latin1]{inputenc}
\usepackage[T1]{fontenc}
\usepackage{calc}
\usepackage{capitalgreekitalic}
\usepackage{url}
\usepackage{enumerate}
\usepackage{color}
\usepackage[tbtags]{amsmath}
\usepackage{amssymb}
\usepackage{upref}
\usepackage{epic,eepic}
\usepackage{times}
\usepackage{dsfont}
\usepackage{comment}
\usepackage{cite}
\usepackage{bbm} 
\usepackage{amsmath}
\usepackage{microtype} 
\usepackage{afterpage}
\usepackage{makecell}
\usepackage{multirow}

\newtheorem{condition}{\bf Condition}
\newtheorem{corollary}{Corollary}

\usepackage{dsfont}

\newcounter{step}
\newlength{\totlinewidth}
\newenvironment{algorithm}{%
  \rule{\linewidth}{1pt}
  \begin{list}{}%
    {\usecounter{step}%
      \settowidth{\labelwidth}{\textbf{Step 2:}}%
      \setlength{\leftmargin}{\labelwidth}%
      \setlength{\topsep}{-2pt}%
      \addtolength{\leftmargin}{\labelsep}%
      \addtolength{\leftmargin}{2mm}%
      \setlength{\rightmargin}{2mm}%
      \setlength{\totlinewidth}{\linewidth}%
      \addtolength{\totlinewidth}{\leftmargin}%
      \addtolength{\totlinewidth}{\rightmargin}%
      \setlength{\parsep}{0mm}%
      \raggedright}}%
  {\end{list}%
  \rule{\linewidth}{1pt}}
\newcounter{substep}

  {\end{list}}

\newlength{\aligntop}
\newlength{\alignbot}
 

\IEEEoverridecommandlockouts

\usepackage{algorithm}
\usepackage{algorithmic}

\makeatletter
\newcommand\semihuge{\@setfontsize\semihuge{19.3}{25}}
\makeatother

\makeatletter
\newcommand\semismall{\@setfontsize\semihuge{12.4}{15}}
\makeatother

\usepackage{subfigure}

\begin{document}

\title{Optimizing Reinforcement Learning Training over Digital Twin Enabled Multi-fidelity Networks 
}

\author{\large{Hanzhi Yu, \textit{Graduate Student Member IEEE}, 
Hasan Farooq, Julien Forgeat, Shruti Bothe,}\\
\large{Kristijonas Cyras, Md Moin Uddin Chowdhury, 
and Mingzhe Chen, \textit{Senior Member IEEE}} \\

\thanks{Hanzhi Yu and Mingzhe Chen are with the Department of Electrical and Computer Engineering, University of Miami, Coral Gables, FL 33146 USA (Emails: \protect\url{{hanzhiyu, mingzhe.chen}@miami.edu}).} 
\thanks{Hasan Farooq, Julien Forgeat, Shruti Bothe, Kristijonas Cyras, and Md Moin Uddin Chowdhury are with the Ericsson Research, Santa Clara, CA, 95054, USA, Emails: \protect\url{{hasan.farooq, julien.forgeat, shruti.bothe, md.moin.uddin.chowdhury}@ericsson.com}. } 
\thanks{Kristijonas Cyras is now with the European AI Office, European Commission, Brussels, Belgium, Email: \protect\url{kristijonas.cyras@ec.europa.eu}. } 
}
\maketitle
%
\begin{abstract}
In this paper, we investigate a novel digital network twin (DNT) assisted deep learning (DL) model training framework. 
In particular, we consider a physical network where a base station (BS) uses several antennas to serve multiple mobile users, and a DNT that is a virtual representation of the physical network. The BS must adjust its antenna tilt angles to optimize the data rates of all users. Due to user mobility, the BS may not be able to accurately track network dynamics such as wireless channels and user mobilities. Hence, a reinforcement learning (RL) approach is used to dynamically adjust the antenna tilt angles. To train the RL, we can use data collected from the physical network and the DNT. The data collected from the physical network is more accurate but incurs more communication overhead compared to the data collected from the DNT. Therefore, it is necessary to determine the ratio of data collected from the physical network and the DNT to improve the training of the RL model. 
We formulate this problem as an optimization problem whose goal is to jointly optimize the tilt angle adjustment policy and the data collection strategy, aiming to maximize the data rates of all users while constraining the time delay introduced by collecting data from the physical network. To solve this problem, we propose a hierarchical RL framework that integrates robust adversarial loss and proximal policy optimization (PPO). The first level robust-RL dynamically adjusts the antenna tilt angles, and the second level PPO determines the ratio of data collected from the physical network to improve the first level robust-RL training performance. Compared to traditional single level RL algorithms such as deep Q network (DQN), the designed method optimizes the data collection ratio and the antenna tilt angles at diverse time intervals, allowing the second level PPO to adjust the data collection ratio with a large timescale using the training information provided by the first level robust-RL. 
Simulation results show that our proposed method reduces the physical network data collection delay by up to 28.01\% and 1$\times$ compared to a hierarchical RL that uses vanilla PPO as the first level RL, and the baseline that uses robust-RL at the first level and selects the data collection ratio randomly. 
\end{abstract}


\section{Introduction}\label{se:intro}
Deep learning (DL) enables wireless network systems to learn from data, adapt to wireless network dynamics (e.g., wireless channel conditions), and make intelligent decisions that enhance network performance and efficiency \cite{8869705, 10198573, 8742579, 9165550, wei2025optimizing, 10891537}. However, deploying deep neural network models over wireless networks requires collecting a large amount of real world data, which is both time and energy consuming \cite{9252924, 10283592}. To address this issue, digital network twin (DNT) technology, which creates a real time virtual representation of physical network environments, has been introduced \cite{9839640}. A DNT can facilitate DL model training by providing simulated network data, thereby reducing the burden of extensive real world data collection \cite{10628026, 9263396}. However, using DNT for training DL models still faces several challenges, such as the analysis of the impacts of inaccurate DNT data on DL model training, optimization of the tradeoff between using the physical network data that is accurate but is time consuming to collect, and using the DNT data that is less time consuming to collect but inaccurate for DL model training \cite{10198573, jia2022accurate, 9833928, 9854866}.

{\color{black}Recent studies have explored the integration of DNT and DL methods for network decision making and system optimization. In \cite{9451579}, the authors proposed a multi agent reinforcement learning (MARL) method to optimize vehicle clustering and computational task offloading in a vehicular network, where an adaptive digital twin (DT) is used to provide real time network information thus improving the MARL performance. In \cite{10235999}, the authors designed a graph attention-based MARL method to jointly optimize task offloading and service caching policies for all users, and used a DNT to provide simulated network state information to support the training of the RL agents. The authors in \cite{10711852} used a DT to provide a pre-initialized policy for online RL in admission control, reducing exploration cost in real network environments. In \cite{10236461}, the authors introduced a DNT assisted federated learning framework and used the DNT to provide real time internet of things (IoT) device information to support IoT device clustering and model training task offloading during the federated learning training. In \cite{10697404}, the authors proposed a DT driven interactive environment in which a multi agent PPO framework learns coordinated beamwidth and tilt control policies online, enabling large scale cellular networks to optimize coverage and throughput without relying on costly real world exploration. However, most of these studies \cite{9451579, 10235999, 10711852, 10236461, 10697404} assumed that the DT can provide accurate network information for DL model training, and do not consider the impacts of inaccurate DT data on model training. Meanwhile, these works \cite{9451579, 10235999, 10711852, 10236461, 10697404} did not consider the tradeoff between collecting data from the physical network and the DT during training a DL model. Since the physical network data and the DNT generated data differ in fidelity and collection cost, understanding the tradeoff between them is crucial for developing practical DL based applications.

Several current studies in \cite{10071987, 10081084, 10286016, 8891763} have investigated how to train a DL model on noisy or uncertain input data. In particular, the authors in \cite{10071987} proposed a framework that injects noise into input data during the model training process to improve the robustness of DL models against data uncertainty caused by channel estimation errors and wireless dynamics (e.g., fading in wireless channels). In \cite{10081084}, the authors proposed a Bayesian DL framework to estimate both data and model uncertainty, making risk-aware decisions for wireless modulation classification and location tasks. The authors in \cite{10286016} proposed a robust channel estimation framework combining a neural network based pilot optimizer and channel estimator, jointly trained by samples generated from properly compensated channel and noise covariance matrices perturbed by modeled estimation errors. Finally, in \cite{8891763}, the authors proposed a neural network based modulation classification framework, that uses a phase-based reordering preprocessing method and a multi-layer long short-term
memory (LSTM) network with attention mechanism to maintain model effectiveness for data with uncertain noise. However, most of these works \cite{10071987, 10081084, 10286016, 8891763} do not consider the communication overhead (e.g., uplink transmission overhead) of collecting accurate and noisy data thus simplifying the model training process. Specifically, one can use DL methods or digital twins to generate noisy data without any data collection overhead. However, for accurate training, we must collect it from the actual network with corresponding communication overhead. Hence, it is necessary to develop a robust training strategy that minimize the communication overhead of accurate data collection, balancing efficiency with model accuracy. }

{\color{black} The main contribution of this work is a novel DNT-assisted wireless DL training framework that allows BSs to dynamically choose training data from both the physical network and the DNT based on network dynamics and training settings, so as to optimize DL model training for physical network performance improvement. 
The key contributions are summarized as follows: 
\begin{itemize}
    \item We consider a DNT enabled network consisting of a physical network with a base station (BS) serving multiple mobile users, and a DNT that is a virtual representation of the physical network. Due to the mobility of the users, the BS cannot accurately capture the dynamics of user movement. Thus, we employ a RL method to dynamically control the antenna tilt angles to maximize user data rates. To train the RL model, the BS can use data collected from either the physical network or the DNT. Data collected from the physical network is more accurate than that from the DNT, but incurs higher communication overhead. Therefore, it is necessary to 
    determine the ratio of data collected from the physical network and the DNT to improve the training of the RL model, so as to optimize the tilt angle of each antenna in response to user mobility. 
    We formulate this problem as an optimization problem aiming to maximize the data rates of all users while considering the delay introduced by collecting data from the physical network. 
    
    \item 
    To address this problem, we introduce a hierarchical RL framework that integrates a robust adversarial loss-RL with a proximal policy optimization (PPO) method. The first level robust-RL dynamically controls the tilt angles of the BS using both noisy and accurate data collected from the DNT and the physical network. The second level PPO adjusts the ratio of the data collected from the physical network and the DNT to support the training of the robust-RL. Compared to traditional policy gradient methods, the robust-RL employs a new loss function that accounts for the worst-case policy. By utilizing the new loss function, the robust-RL enhances model robustness against data noise caused by DNT, allowing more DNT data to be used for model training and reducing the overhead of collecting data from a physical network. Compared to traditional single layer RL algorithms, the designed hierarchical RL method optimizes different optimization variables (i.e., the data collection ratio and the antenna tilt angles) at different temporal resolutions by allowing the second level PPO to optimize long-term strategic parameters (i.e., the data collection ratio), while the first level policy focuses on short-term operational decisions (i.e., the antenna tilt angles). 

    \item We analyze the convergence of the second level RL in our proposed hierarchical RL framework. The analytical result shows that the second level RL converges to approximate stationarity in expectation. 
    
\end{itemize}}
Simulation results show that our proposed method reduces the physical network data collection delay by up to 28.01\% and 1$\times$ compared to a hierarchical RL that uses vanilla PPO as the first level RL, and the baseline that uses robust-RL at the first level and selects the data collection ratio randomly. 

The rest of this paper is organized as follows. The system model and problem formulation are introduced in Section \ref{se:system}. The design of the hierarchical RL framework is introduced in Section \ref{se:solution}. The convergence of our proposed method is analyzed in Section \ref{se:analysis}. Simulation results are presented and discussed in Section \ref{se:simulation}. Finally, conclusions are drawn in Section \ref{se:conclusion}.

\section{System Model and Problem Formulation}\label{se:system}
{\color{black} We consider a DNT enabled cellular network consisting of: 1) a physical network including a BS that 
serve a set $\mathcal{U}$ of $U$ mobile users in a set $\mathcal{C}$ of $C$ cells, 
and 2) a DNT of the physical network that is generated and controlled by a cloud server, as it shows in Fig. \ref{fig:system_model}. In the considered model, the users in each cell are served by an antenna of the BS such that the BS totally has $C$ antennas. \cite{bouton2021coordinated} The BS must adjust its antenna tilt angles to optimize the data rate of mobile users. Due to the mobility of the users, the BS cannot accurately learn the dynamics of user mobility. Thus, we utilize a RL method to dynamically control the antenna tilt angles. To train the RL model, the BS can use the data collected from the physical network and the DNT. Here, the data collected from the physical network is more accurate compared to the data collected from the DNT. However, collecting data from the physical network will introduce more communication overhead. Therefore, it is necessary to optimize the RL training data collection strategy (e.g., from the physical network or the DNT).} Next, we first introduce the mobility model of the users. Then, we describe the communication transmission model. Subsequently, we introduce the RL method that the BS uses to control the tilt angles. Given the defined physical network and the implemented RL, we introduce the model of the DNT. Finally, we formulate the optimization problem of RL training via data collection. 

\begin{figure}[t]
  \begin{center}
    \includegraphics[width= 0.95\linewidth]{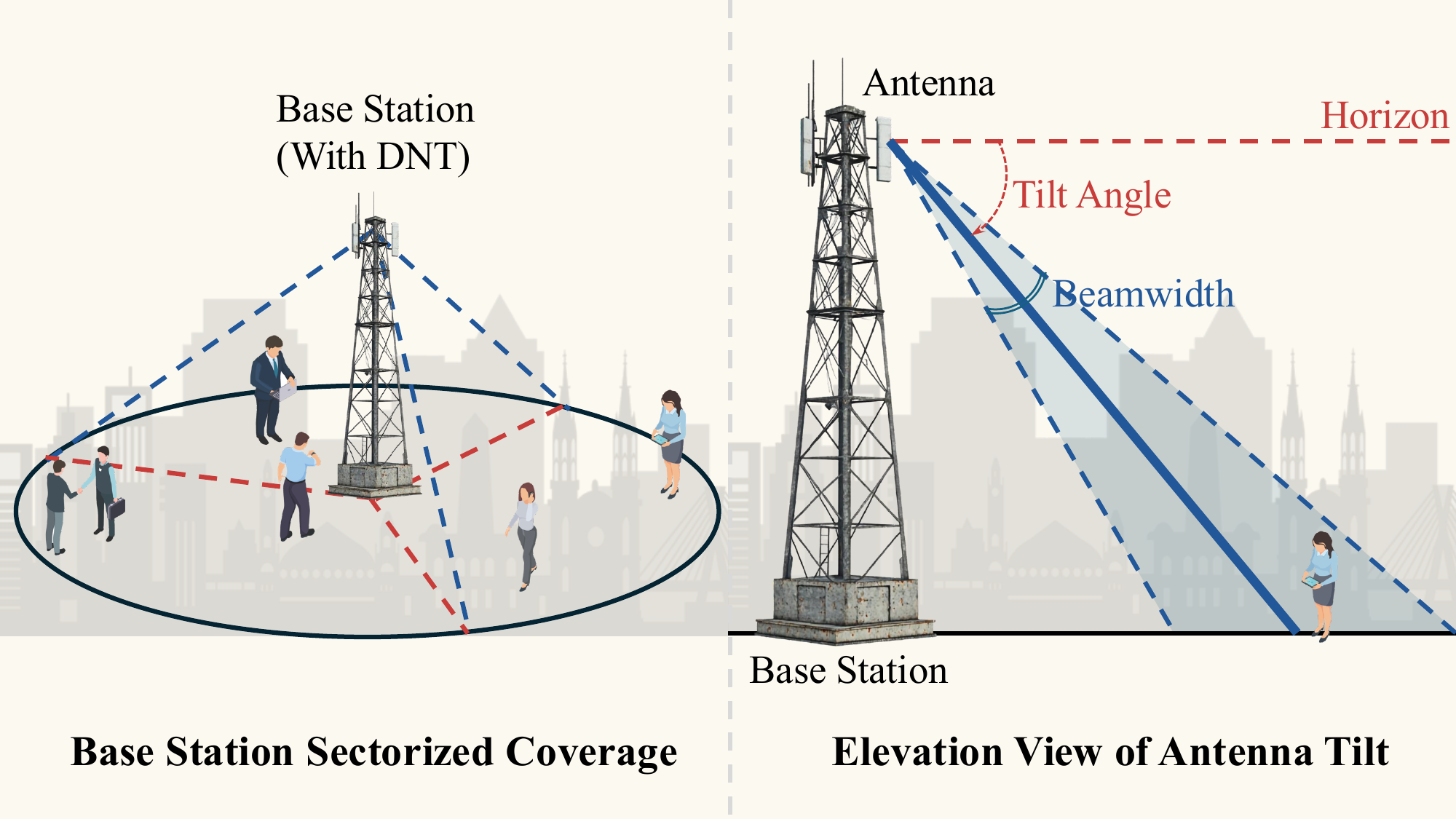}
    \caption{The considered DNT enabled cellular network. } 
    \label{fig:system_model}
    \vspace{-0.5cm}
  \end{center}
\end{figure}

\subsection{Mobility Model}
We use a random walk process to model the mobility of each user \cite{tabassum2019fundamentals}. {\color{black} This model is used as an example, and the proposed framework can be applied to other mobility models.} {\color{black} At each time slot $t$, user $u$ can choose from five possible movements: 1) remain at the current position, 2) move forward, 3) move backward, 4) move left, or 5) move right.} The probability of each movement is represented by $\boldsymbol{p}_u = \left[ p_{u,1}, p_{u,2}, p_{u,3}, p_{u,4}, p_{u,5} \right]$. Let $\boldsymbol{l}^{\textrm{U}}_{u,t} = \left[ l^{\textrm{U}}_{u,t,1}, l^{\textrm{U}}_{u,t,2} \right]$ be the user's position at time slot $t$. Then, the position of user $u$ at time slot $t+1$ is determined as follows: 
\begin{equation}\label{eq:move}
     \boldsymbol{l}_{u, t+1}^\textrm{U} = \left \{
        \begin{aligned}
            &\left[ l_{u,t,1}^\textrm{U}, l_{u,t,2}^\textrm{U} \right] &&\text{with probability } p_{u,1}, \\
            &\left[ l_{u,t,1}^\textrm{U}, l_{u,t,2}^\textrm{U} + \Delta l \right] &&\text{with probability } p_{u,2}, \\
            &\left[ l_{u,t,1}^\textrm{U}, l_{u,t,2}^\textrm{U} - \Delta l \right] &&\text{with probability } p_{u,3}, \\
            &\left[ l_{u,t,1}^\textrm{U} - \Delta l, l_{u,t,2}^\textrm{U} \right] &&\text{with probability } p_{u,4}, \\
            &\left[ l_{u,t,1}^\textrm{U} + \Delta l, l_{u,t,2}^\textrm{U} \right] &&\text{with probability } p_{u,5}, \\
        \end{aligned}
    \right.
\end{equation}
where $\Delta l$ is the distance that a user moves in a single time slot. 

\subsection{Transmission Model}
The BS serves the users and collect physical network data for RL training. We first discuss the downlink scenario. At each time slot $t$, a user is considered to be served by antenna sector $c$ if the user is located within the corresponding coverage region of antenna sector $c$. Thus, the signal-to-interference-plus-noise ratio (SINR) of user $u$ that is associated with cell $c$ of the BS is given by \cite{asghar2018concurrent} 
\begin{equation}\label{eq:sinr_down}
    \gamma_{uc,t}^\textrm{D} \left( \boldsymbol{\psi}_t^\textrm{T} \right) = \frac{P_0 g_{uct}^\textrm{B} g_{ut}^\textrm{U} \delta_{uct} a \left( d_{uct} \right)^{-\beta}}{N_0 + \sum_{\forall i \in \mathcal{C}/c} P_0 g_{uit}^\textrm{B} g_{ut}^\textrm{U} \delta_{uit} a \left( d_{uit} \right)^{-\beta}}, 
\end{equation}
where $\boldsymbol{\psi}_t^\textrm{T} = \left[ \psi_{1,t}^\textrm{T}, ..., \psi_{C,t}^\textrm{T} \right]$ is a vector of tilt angles of all antennas, with $\psi_{c,t}^\textrm{T}$ being the tilt angle of the antenna in cell $c$; $P_0$ is the transmit power of the BS; $g_{ut}^\textrm{U}$ is the antenna gain of user $u$; $\delta_{uct}$ is the large scale shadowing of the signal transmitted from cell $c$ to user $u$; $a$ is the path loss; $d_{uct}$ is the distance between user $u$ and cell $c$; $\beta$ is the path loss exponent; $N_0$ is the noise power; and $g_{uct}^\textrm{B}$ is the antenna gain of cell $c$ towards user $u$. We assume that each user is equipped with a omnidirectional antenna and the BS is equipped with the directional antennas, $g_{uct}^\textrm{B}$ is expressed as \cite{asghar2018concurrent}
\begin{equation}\label{eq:gain_b}
     g_{uct}^\textrm{B} = 10^{-1.2 \left( \lambda^\textrm{V} \left( \frac{\psi_{uct} - \psi_{ct}^\textrm{T}}{\theta^\textrm{V}} \right)^2 + \lambda^\textrm{H} \left( \frac{\phi_{uct} - \phi_{ct}^\textrm{A}}{\theta^\textrm{H}} \right)^2 \right)},
\end{equation}
where $\lambda^\textrm{V} > 0$ and $\lambda^\textrm{H} > 0$ are the weight parameters of the horizontal and vertical beam patterns of the antenna, $\psi_{uct}$ and $\phi_{uct}$ are the vertical and horizontal angle between user $u$ and the antenna of cell $c$, $\phi_{ct}^\textrm{A}$ is the azimuth of the antenna of cell $c$, and $\theta^\textrm{V}$, $\theta^\textrm{H}$ are the vertical and horizontal beamwidth of the antenna of cell $c$. Based on (\ref{eq:sinr_down}) and (\ref{eq:gain_b}), the data rate of user $u$ associated with cell $c$ at time slot $t$ is 
\begin{equation}\label{eq:r_down}
    r_{uc,t}^\textrm{D} \left( \boldsymbol{\psi}_t^\textrm{T} \right) = W \log_2 \left( 1 + \gamma_{uc,t}^\textrm{D} \left( \boldsymbol{\psi}_t^\textrm{T} \right) \right), 
\end{equation}
where $W$ is the bandwidth. 

Similarly, the SINR of the signal received by antenna of cell $c$ from user $u$ is 
\begin{equation}\label{eq: }
    \gamma_{uc,t}^\textrm{U} \left( \boldsymbol{\psi}_t^\textrm{T} \right) = \frac{P_u g_{uct}^\textrm{B} g_{ut}^\textrm{U} \delta_{uct}^\textrm{U} a \left( d_{uct} \right)^{-\beta}}{N_0 + \sum_{\forall i \in \mathcal{C}/c} P_i g_{uit}^\textrm{B} g_{ut}^\textrm{U} \delta_{uit}^\textrm{U} a \left( d_{uit} \right)^{-\beta}}, 
\end{equation}
where $P_u$ is the transmit power of user $u$, and $\delta_{uct}^\textrm{U}$ is shadowing effect of the signal transmitted from user $u$ to cell $c$. Thus, the data rate of user $u$ transmitting data to the antenna of cell $c$ is 
\begin{equation}
    r_{uc,t}^\textrm{U} \left( \boldsymbol{\psi}_t^\textrm{T} \right) = W \log_2 \left( 1 + \gamma_{uc,t}^\textrm{U} \left( \boldsymbol{\psi}_t^\textrm{T} \right) \right).  
\end{equation}

\subsection{RL for Tilt Angle Adjustment}
To optimize the data rate of the users, we must adjust the tilt angles of the antennas based on the mobility of the users. We consider a discrete-time system, where a time period is divided into time slots, and the network statistics are assumed to remain constant within each time slot. {\color{black} In our considered tilt angle adjustment problem, the BS adjusts the antenna tilt angles every $N$ time slots.} Therefore, the BS does not know the future positions of the users, making the tilt angle adjustment problem difficult to be solved by traditional optimization methods. Here, we assume that the BS uses an RL method to address this tilt angle adjustment problem. \cite{8789640, 10587126}
To train the RL method, the BS needs to collect a series of 
data from the system. We define each transition 
that consists of the state, the action, and the reward of the RL as a basic unit of data collection. 

\subsection{DNT Model}
The data generated by DNT may contain errors due to wireless data transmission and inaccurate DNT synchronization \cite{10906634}. 
Then, a user location $\hat{\boldsymbol{l}}_{u,t}^\textrm{U}$ generated by the DNT 
can be expressed as 
$\hat{\boldsymbol{l}}_{u,t}^\textrm{U} = \boldsymbol{l}_{u,t}^\textrm{U} + \epsilon \left( \boldsymbol{l}_{u,t}^\textrm{U} \right)$, where $\epsilon \left( \boldsymbol{l}_{u,t}^\textrm{U} \right)$ is the DNT data generation error. 
To train the RL method, the BS can either use the data generated by the DNT or the transitions collected from the physical network. Let $e$ be an index of a epoch that consists of several time slots collected for training the RL. We assume that a set $\mathcal{B}_e$ of $\left| \mathcal{B}_e \right|$ transitions are collected to train the RL in epoch $e$, $\mathcal{B}_e^\textrm{P}$ is the subset of $\left| \mathcal{B}_e^\textrm{P} \right|$ transitions collected from the physical network, and $\mathcal{B}_e^\textrm{D}$ is the subset of $\left| \mathcal{B}_e^\textrm{D} \right|$ transitions collected from the DNT, such that $\left| \mathcal{B}_e \right| = \left| \mathcal{B}_e^\textrm{P} \right| + \left| \mathcal{B}_e^\textrm{D} \right|$. We use $\rho_e \in \left[ 0,1 \right]$ to denote the ratio between $\left| \mathcal{B}_e^\textrm{P} \right|$ and $\left| \mathcal{B}_e \right|$ (i.e., $\rho_e = \frac{\left| \mathcal{B}_e^\textrm{P} \right|}{\left| \mathcal{B}_e \right|}$). 
The collection of each transition $b$ of $\mathcal{B}_e^\textrm{P}$ requires $N$ time slots, since the BS changes the antenna tilt angles every $N$ time slots. We denote $t_{e,b}$ as the first time slot that collects transition $b$ in epoch $e$, and $t_{e,b} + N - 1$ as the last time slot. Let $D$ represents the size of data that each user transmits to the BS to generate one transition and update the DNT. Then, the transmission delay of collecting each transition $\tau_{e,b} $ is 
\begin{equation}\label{eq: tau_e_b}
    \tau_{e,b} \left( \boldsymbol{\psi}^\textrm{T}_{t_{e,b}} \right) = \sum_{t = t_{e,b}}^{t_{e,b} + N - 1} \max \frac{D}{r_{uc,t}^\textrm{U} \left( \boldsymbol{\psi}^\textrm{T}_{t_{e,b}} \right)}. 
\end{equation}

\subsection{Problem Formulation}
Given our designed system model, our goal is to maximize the data rate of all users, while meeting the data collection time constraint. 
Our optimization problem includes optimizing ratio $\rho_e$ over a set $\mathcal{E}$ of $E$ training epoch, and the tilt angle $\boldsymbol{\psi}^T_t$ of all antennas over all sets $\mathcal{T}_e$ of $T_e$ time slots. The optimization problem is formulated as 
\begin{subequations}\label{eq:problem}
    \begin{equation}\tag{\theequation}\label{eq:opt_prob}
    \begin{split}
        \max_{\boldsymbol{\psi}^\textrm{T}_t, \rho_e, \forall t \in \mathcal{T}, e \in \mathcal{E}} \sum_{e=1}^E 
        \sum_{t \in \mathcal{T}_e} \sum_{u=1}^U r^\textrm{D}_{uc,t} \left( \boldsymbol{\psi}^\textrm{T}_t \right)
    \end{split}
    \end{equation}
    \begin{flalign}
        &&\text{s.t. } &
        \rho_e  \in \left[ 0,1 \right], e \in \mathcal{E}, && \label{eq:const1} \\
        && & \sum_{b=1}^{\rho_e \left| \mathcal{B}_e \right|} \tau_{e,b} \left( \boldsymbol{\psi}^\textrm{T}_{t_{e,b}} \right) \leq \tau_{max}, e \in \mathcal{E}, && \label{eq:const2} \\
        && & \psi^\textrm{T}_{min} \leq \psi^\textrm{T}_{c,t} \leq \psi^\textrm{T}_{max}, \forall c \in \mathcal{C}, t \in \mathcal{T}_e, e \in \mathcal{E}, && \label{eq:const3}
    \end{flalign}
\end{subequations}
where 
$\tau_{max}$ is the maximum data collection time that the system allows per epsiode, and $\psi^\textrm{T}_{min}, \psi^\textrm{T}_{max}$ are the minimum and maximum value of the tilt angle. 
In problem (\ref{eq:problem}), constraint (\ref{eq:const1}) balances the transitions collected from the physical network and the DNT. {\color{black} Constraint (\ref{eq:const2}) is a constraint of the data collection time of epoch $e$.} Constraint (\ref{eq:const3}) indicates that the BS can only adjust the antenna angle within a certain range. 

{\color{black} Problem (\ref{eq:opt_prob}) is challenging to solve due to the following reasons. First, the relationship between the data collection ratio $\rho_e$ and the sum of data rates of all users over $T_e$ time slots
cannot be explicitly expressed, making it difficult to assess how the data collection policy impacts the RL training performance. Second, the relationship between the data collection time $\tau_e \left( \boldsymbol{\psi}^\textrm{T}_{t_{e,b}} \right)$ and the ratio $\rho_e$ is non-linear and this relationship will also be influenced by unpredictable network conditions, such as the noise power $N_0$. Third, the data collection ratio $\rho_e$ and the tilt angle of all antennas $\boldsymbol{\psi}_t^\textrm{T}$ are coupling since $\rho_e$ affects the training performance of the RL that determines $\boldsymbol{\psi}_t^\textrm{T}$. Finally, the data collected from the DNT contain errors, and the BS cannot know the error distributions, which increases the complexity of solving problem (\ref{eq:problem}).}
To solve problem (\ref{eq:opt_prob}), we propose a hierarchical RL that enables the BS to optimize tilt angle over $N$ steps, and data collection ratio over $E$ training epochs by learning the implicit relationship of $\rho_e$, $\boldsymbol{\psi}_t^\textrm{T}$, and the sum of data rates of all users. 

\section{Proposed Hierarchical Reinforcement Learning}\label{se:solution}
To solve problem (\ref{eq:problem}), we introduce a hierarchical RL framework that integrates robust adversarial loss-RL \cite{NEURIPS2021_dbb42293} with proximal policy optimization (PPO). {\color{black} The robust-RL is used to determine the tilt angle $\boldsymbol{\psi}^\textrm{T}_t$ of all antennas at one step that consists of $N$ time slots.} The PPO is used to optimize the data collection ratio $\rho_e$ for each epoch $e$ using the training information of the first RL.
Compared to traditional single RL algorithms such as deep Q network (DQN), the designed hierarchical RL method optimizes the data collection ratio and the antenna tilt angles at diverse time intervals, allowing the second level PPO to adjust the data collection ratio with a large timescale using the training information provided by the first level robust-RL, and allowing the first level robust-RL focuses on adjusting the antenna tilt angles with a small timescale. 
Compared to traditional policy gradient methods, the robust-RL employs a new loss function that considers the worst-case policy noisy data is used for RL training. By utilizing the new loss function, the robust-RL enhances model robustness against noise in the training data generated by DNT thus efficiently utilizing more DNT data for its model training thus reducing the overhead of collecting data from a physical network. 
Next, we first introduce the components of the first level robust-RL for tilt angle adjustment. Then, we explain the procedure of training the robust-RL. Subsequently, we introduce the components of the second level PPO for data collection ratio optimization. Finally, we analyze the total data collection time and its connection to the reward function of the second level PPO. 
The structure of the proposed method is shown in Fig. \ref{fig:structure}. 
\begin{figure}[t]
  \begin{center}
    \includegraphics[width=9.5cm]{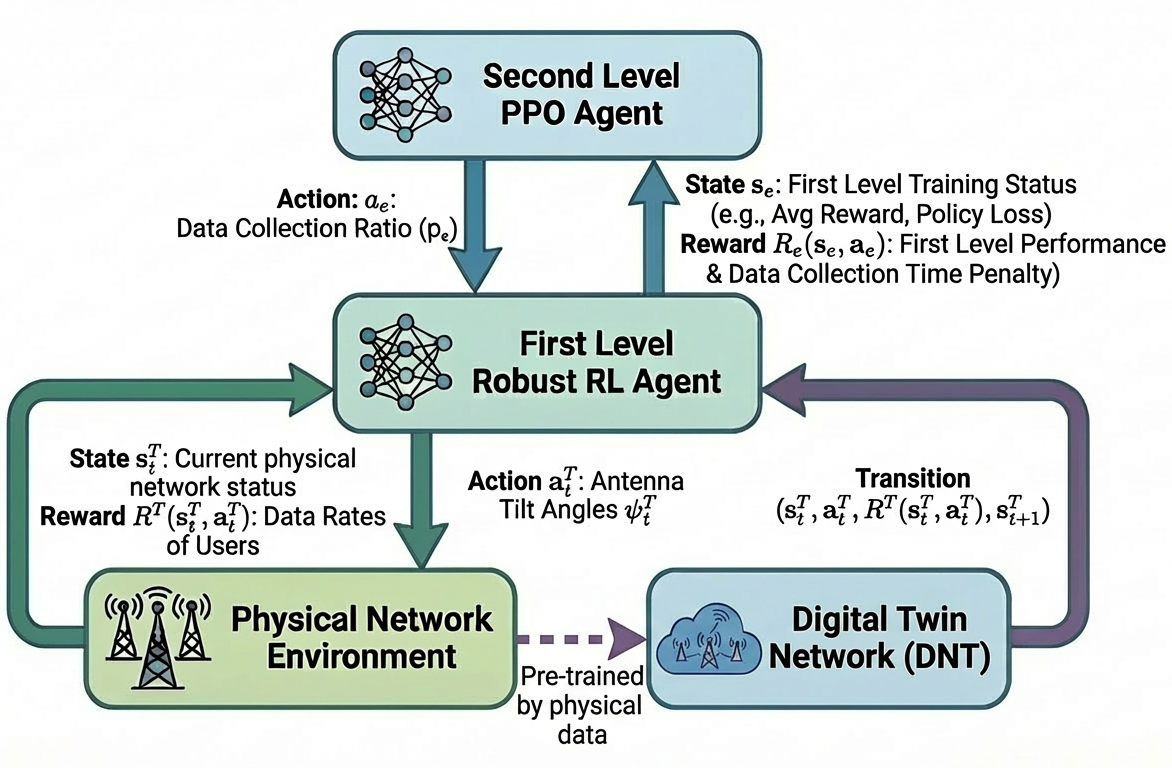}
    \caption{The structure of the proposed hierarchical RL framework. } 
    \label{fig:structure}
    \vspace{-0.5cm}
  \end{center}
\end{figure}

\subsection{The Components of the First Level Robust-RL}
The first level robust-RL has the following seven components: 

\subsubsection{Agent} The agent is the BS. It observes the physical network status and adjust the antenna tilt angles of the BS. 

\subsubsection{State} The state of the BS describes the current physical network status or the DNT status (if the data is collected from the DNT). Here, we use the positions of all users to represent the physical network status or the DNT status. In particular, if the positions of all users are observed from the physical network, the state of the BS at time slot $t$ is $\boldsymbol{s}_t^\textrm{T} = \left[ \boldsymbol{l}_{1,t}^\textrm{U}, ..., \boldsymbol{l}_{U,t}^\textrm{U} \right]$. Otherwise, if the positions of all users are observed from the DNT, $\boldsymbol{s}_t^\textrm{T} = \left[ \hat{\boldsymbol{l}}_{1,t}^\textrm{U}, ..., \hat{\boldsymbol{l}}_{U,t}^\textrm{U} \right]$.  

\subsubsection{Action} The action of the BS determines the antennas tilt angle of the BS. Thus, each action of the BS at time slot $t$ is $\boldsymbol{a}_t^\textrm{T} = \boldsymbol{\psi}_t^\textrm{T}$. If the antenna tilt angles are collected from the DNT, $\boldsymbol{a}_t^\textrm{T} = \hat{\boldsymbol{\psi}}_t^\textrm{T}$

\subsubsection{Reward function} The reward function $R^\textrm{T} \left( \boldsymbol{s}_t^\textrm{T}, \boldsymbol{a}_t^\textrm{T} \right)$ evaluates action $\boldsymbol{a}_t^\textrm{T}$ based on state $\boldsymbol{s}_t^\textrm{T}$. Since the BS adjusts the antenna tilt angles every $N$ time slots, the reward function of the BS is the sum of the data rates of all users over $N$ time slots, which is given by 
 \vspace{-0.1cm}
\begin{equation}\label{eq:reward_tilt}
    R^\textrm{T} \left( \boldsymbol{s}_t^\textrm{T}, \boldsymbol{a}_t^\textrm{T} \right) = \left \{ 
    \begin{aligned}
        &\sum_{n=t}^{t+N} \sum_{u=1}^U r^\textrm{D}_{uc,n} \left( \boldsymbol{\psi}_t^\textrm{T} \right) - \epsilon \left( R^\textrm{T} \right), \\
        & \quad \quad \text{if state $\boldsymbol{s}^\textrm{T}_t$ is collected from the DNT, } \\
        &\sum_{n=t}^{t+N} \sum_{u=1}^U r^\textrm{D}_{uc,n} \left( \boldsymbol{\psi}_t^\textrm{T} \right), \text{otherwise, }
    \end{aligned}
    \right.
\end{equation}
where $\epsilon \left( R^\textrm{T} \right)$ is the error associated with the reward. 

\subsubsection{Value function} The value function estimates the expected reward of the BS at each state $\boldsymbol{s}^\textrm{T}_t$. The BS uses a deep neural network (DNN) with parameter $\boldsymbol{W}^\textrm{V}$ to approximate the value function $V_{\boldsymbol{W}^\textrm{V}} \left( \boldsymbol{s}^\textrm{T}_t \right)$. 

\subsubsection{Advantage function} The advantage function describes the improvement in reward when taking action $\boldsymbol{a}^\textrm{T}_t$ in state $\boldsymbol{s}^\textrm{T}_t$ compared to the expected reward in that state. The advantage function of the BS is implicitly parameterized by the value function parameters, which is given by 
\begin{equation}\label{eq:advantage_tilt}
    A_{\boldsymbol{W}^\textrm{V}} \left( \boldsymbol{s}^\textrm{T}_t, \boldsymbol{a}^\textrm{T}_t \right) = R^\textrm{T} \! \left( \! \boldsymbol{s}_t^\textrm{T}, \boldsymbol{a}_t^\textrm{T} \! \right) \! + \! \lambda^\textrm{T} V_{\boldsymbol{W}^\textrm{V}} \left( \! \boldsymbol{s}^\textrm{T}_{t+1} \! \right) - V_{\boldsymbol{W}^\textrm{V}} \left( \! \boldsymbol{s}^\textrm{T}_t \! \right), 
\end{equation}
where $\lambda^\textrm{T} \in \left( 0,1 \right)$ is the discount factor. 

\begin{table}[t]
    \centering
    \caption{Summary of DNT Data}
    \begin{tabular}{|c|c|c|}
        \hline
        \textbf{Data Type} & \textbf{Physical Network Data} & \textbf{DNT Data} \\ \hline        
        
        \multirow{5}{*}{Static Data} 
        & $\mathcal{U}, U$ & $\mathcal{U}, U$ \\ \cline{2-3} 
        & $\mathcal{C}, C$ & $\mathcal{C}, C$ \\ \cline{2-3} 
        & $P$ & $P$ \\ \cline{2-3} 
        & $\theta^\textrm{V}, \theta^\textrm{H}$ & $\theta^\textrm{V}, \theta^\textrm{H}$ \\ \cline{2-3} 
        & $\lambda^\textrm{V}, \lambda^\textrm{H}$ & $\lambda^\textrm{V}, \lambda^\textrm{H}$ \\
        \hline
        
        \multirow{7}{*}{Runtime Data} 
        & $\boldsymbol{l}_{1,t}^{\textrm{U}}, ..., \boldsymbol{l}_{U,t}^{\textrm{U}}$ & $\hat{\boldsymbol{l}}_{1,t}^{\textrm{U}}, ..., \hat{\boldsymbol{l}}_{U,t}^{\textrm{U}}$ \\ \cline{2-3} 
        & $\psi_{uct}, \phi_{uct}, \psi_{ct}^\textrm{T}, \phi_{ct}^\textrm{A}$ & $\hat{\psi}_{uct}, \hat{\phi}_{uct}, \hat{\psi}_{ct}^\textrm{T}, \hat{\phi}_{ct}^\textrm{A}$ \\ \cline{2-3} 
        & $g_{ut}^\textrm{U}$ & $\hat{g}_{ut}^\textrm{U}$ \\ \cline{2-3} 
        & $\sigma$ & $\hat{\sigma}$ \\ \cline{2-3} 
        & $\gamma_{1,t}, ..., \gamma_{U,t}$ & $\hat{\gamma}_{1,t}, ..., \hat{\gamma}_{U,t}$ \\ \cline{2-3} 
        & $r_{1,t}, ..., r_{U,t}$ & $\hat{r}_{1,t}, ..., \hat{r}_{U,t}$ \\ \cline{2-3} 
        & $R^\textrm{T} \left( \boldsymbol{s}_t^\textrm{T}, \boldsymbol{a}_t^\textrm{T} \right)$ & $\hat{R}_t^\textrm{T}$ \\
        \hline
    \end{tabular}
    \vspace{-0.3cm}
    \label{tab:dnt}
\end{table}

\subsubsection{Policy} 
The policy of the BS is the conditional probability of the BS selecting action $\boldsymbol{a}^\textrm{T}_t$ given state $\boldsymbol{s}^\textrm{T}_t$. The policy is represented by a DNN parameterized by $\boldsymbol{W}^\textrm{P}$ and is denoted as $\pi_{\boldsymbol{W}^\textrm{P}} \left( \boldsymbol{a}^\textrm{T}_t | \boldsymbol{s}^\textrm{T}_t \right)$. In our considered problem, the policy may be affected by the noise in the data collected from the DNT. {\color{black} To improve the algorithm's robustness against the noise, we introduce the worst-case policy as 
\begin{equation}\label{eq:worst_case_policy}
    \hat \pi_{\boldsymbol{W}^\textrm{P}} \left( \boldsymbol{a}^\textrm{T}_t | \boldsymbol{s}^\textrm{T}_t; \rho_e \right) = \left \{
    \begin{aligned}
        \underline{\pi}_{\boldsymbol{W}^\textrm{P}} \left( \boldsymbol{a}^\textrm{T}_t | \boldsymbol{s}^\textrm{T}_t; \rho_e \right), \, & \text{if} A_{\boldsymbol{W}^\textrm{V}} \left( \boldsymbol{s}^\textrm{T}_t, \boldsymbol{a}^\textrm{T}_t \right) \geq 0; \\
        \overline{\pi}_{\boldsymbol{W}^\textrm{P}} \left( \boldsymbol{a}^\textrm{T}_t | \boldsymbol{s}^\textrm{T}_t; \rho_e \right), \, & \text{otherwise},     
    \end{aligned}
    \right. 
\end{equation}
where $\underline{\pi}_{\boldsymbol{W}^\textrm{P}} \left( \boldsymbol{a}^\textrm{T}_t | \boldsymbol{s}^\textrm{T}_t; \rho_e \right)$ is a lower bound probability of the BS selecting the positive-advantage action $\boldsymbol{a}^\textrm{T}_t$ given state $\boldsymbol{s}^\textrm{T}_t$ and the data collection ratios $\rho_e$, 
and $\overline{\pi}_{\boldsymbol{W}^\textrm{P}} \left( \boldsymbol{a}^\textrm{T}_t | \boldsymbol{s}^\textrm{T}_t; \rho_e \right)$ is an upper bound probability of the BS selecting the negative-advantage action $\boldsymbol{a}^\textrm{T}_t$ given state $\boldsymbol{s}^\textrm{T}_t$ and the data collection ratio $\rho_e$. Assuming that the original policy follows the multivariate Gaussian distribution, then, $\underline{\pi}_{\boldsymbol{W}^\textrm{P}} \left( \boldsymbol{a}^\textrm{T}_t | \boldsymbol{s}^\textrm{T}_t; \rho_e \right)$ and $\overline{\pi}_{\boldsymbol{W}^\textrm{P}} \left( \boldsymbol{a}^\textrm{T}_t | \boldsymbol{s}^\textrm{T}_t; \rho_e \right)$ is given by 
\begin{equation}\label{eq:overline_underline_pi}
\begin{split}
    & \underline{\pi}_{\boldsymbol{W}^\textrm{P}} \left( \boldsymbol{a}^\textrm{T}_t | \boldsymbol{s}^\textrm{T}_t; \rho_e \right) = \frac{e^{-\overline{d}/2}}{\left( \left( 2\pi \right)^{k/2} \det \boldsymbol{\Sigma} \right)^{0.5}}, \\
    & \overline{\pi}_{\boldsymbol{W}^\textrm{P}} \left( \boldsymbol{a}^\textrm{T}_t | \boldsymbol{s}^\textrm{T}_t; \rho_e \right) = \frac{e^{-\underline{d}/2}}{\left( \left( 2\pi \right)^{k/2} \det \boldsymbol{\Sigma} \right)^{0.5}}, 
\end{split}
\end{equation}
where $\boldsymbol{\Sigma}$ is the covariance matrix of the Gaussian policy, $k$ is the dimension of the action space, and $\overline{d}, \underline{d}$ represent the maximum and minimum Mahalanobis distances between the action and the policy mean, considering the possible input noises related to $\rho_e$. $\overline{d}, \underline{d}$ can be expressed as 
\begin{equation}\label{eq:m_dis}
\begin{split}
    & \overline{d} = \max_{\boldsymbol{\mu} \in \left[ \underline{\boldsymbol{\mu}}, \overline{\boldsymbol{\mu}} \right]} \left( \boldsymbol{a}^\textrm{T}_t - \boldsymbol{\mu} \right)^T \boldsymbol{\Sigma}^{-1} \left( \boldsymbol{a}^\textrm{T}_t - \boldsymbol{\mu} \right), \\
    & \underline{d} = \min_{\boldsymbol{\mu} \in \left[ \underline{\boldsymbol{\mu}}, \overline{\boldsymbol{\mu}} \right]} \left( \boldsymbol{a}^\textrm{T}_t - \boldsymbol{\mu} \right)^T \boldsymbol{\Sigma}^{-1} \left( \boldsymbol{a}^\textrm{T}_t - \boldsymbol{\mu} \right),
\end{split}
\end{equation}
where $\boldsymbol{\mu}$ is the mean vector of the Gaussian policy, and $\left[ \underline{\boldsymbol{\mu}}, \overline{\boldsymbol{\mu}} \right]$ the range of possible values that the policy mean $\boldsymbol{\mu}$ can take when $\boldsymbol{s}^\textrm{T}_t$ is affected by noise related to the data collection ratio $\rho_e$. In (\ref{eq:worst_case_policy}), when $A_{\boldsymbol{W}^\textrm{V}} \left( \boldsymbol{s}^\textrm{T}_t, \boldsymbol{a}^\textrm{T}_t \right) \geq 0$, indicating action $\boldsymbol{a}^\textrm{T}_t$ is beneficial based on state $\boldsymbol{s}^\textrm{T}_t$, the worst-case policy is given by $\underline{\pi}_{\boldsymbol{W}^\textrm{P}} \left( \boldsymbol{a}^\textrm{T}_t | \boldsymbol{s}^\textrm{T}_t; \rho_e \right)$ as it provides the minimum possible probability of the BS selecting the good action $\boldsymbol{a}^\textrm{T}_t$. Similarly, when $A_{\boldsymbol{W}^\textrm{V}} \left( \boldsymbol{s}^\textrm{T}_t, \boldsymbol{a}^\textrm{T}_t \right) < 0$, indicating the action is unfavorable, the worst-case policy is given by $\overline{\pi}_{\boldsymbol{W}^\textrm{P}} \left( \boldsymbol{a}^\textrm{T}_t | \boldsymbol{s}^\textrm{T}_t; \rho_e \right)$ as it provides the maximum possible probability of the BS selecting the bad action $\boldsymbol{a}^\textrm{T}_t$.} 


\subsection{Training of the First Level Robust-RL}
Next, we introduce the training process of the robust-RL. At each training epoch $e$, the BS collects a set of transitions from the physical network and the DNT based on the data collection ratio $\rho_e$. The collected transitions are saved in a replay buffer. Then, a batch $\mathcal{D}^\textrm{T}_e$ of transitions are sampled from the replay buffer to update the policy network $\boldsymbol{W}^\textrm{P}$ and the value network $\boldsymbol{W}^\textrm{V}$. \cite{lin1992self} The loss function of updating the policy network $\boldsymbol{W}^\textrm{P}$ consists of: 1) the standard PPO loss function that measures the standard performance of the model, and 2) the adversarial loss function that captures the robust performance of the model. The standard PPO loss function is expressed as
\begin{equation}\label{eq:loss_tilt_norm}
\begin{split}
    \mathcal{L}^\textrm{N} \left( \boldsymbol{W}^\textrm{P}_e \right) = \mathbb{E}_{b \sim \mathcal{D}^\textrm{T}_e} \! \left[ \! - \! \min \! \left( \! \frac{\pi_{\boldsymbol{W}^\textrm{P}_e} \! \left( \boldsymbol{a}^\textrm{T}_t | \boldsymbol{s}^\textrm{T}_t \right)}{\pi_{\boldsymbol{W}^\textrm{P}_{e-1}} \! \left( \boldsymbol{a}^\textrm{T}_t | \boldsymbol{s}^\textrm{T}_t \right)} A_{\boldsymbol{W}^\textrm{V}} \! \left( \! \boldsymbol{s}^\textrm{T}_t, \boldsymbol{a}^\textrm{T}_t \!\right), \right. \right. \\ \left. \left. \textrm{clip} \! \left( \! \frac{\pi_{\boldsymbol{W}^\textrm{P}_e} \! \left( \boldsymbol{a}^\textrm{T}_t | \boldsymbol{s}^\textrm{T}_t \right)}{\pi_{\boldsymbol{W}^\textrm{P}_{e-1}} \! \left( \boldsymbol{a}^\textrm{T}_t | \boldsymbol{s}^\textrm{T}_t \right)}, 1-\eta, 1+\eta \right) \! A_{\boldsymbol{W}^\textrm{V}} \! \left( \! \boldsymbol{s}^\textrm{T}_t, \boldsymbol{a}^\textrm{T}_t \! \right) \! \right) \right], 
\end{split}
\end{equation}
where $\textrm{clip} \left( \cdot \right)$ restricts the policy ratio update to a pre-defined range $\left[ 1-\eta, 1+\eta \right]$, with $\eta \in \left( 0,1 \right)$, $\frac{\pi_{\boldsymbol{W}^\textrm{P}_e} \left( \boldsymbol{a}^\textrm{T}_t | \boldsymbol{s}^\textrm{T}_t \right)}{\pi_{\boldsymbol{W}^\textrm{P}_{e-1}} \left( \boldsymbol{a}^\textrm{T}_t | \boldsymbol{s}^\textrm{T}_t \right)}$ is the ratio between the current policy and the policy from the previous epoch. In (\ref{eq:loss_tilt_norm}), $\frac{\pi_{\boldsymbol{W}^\textrm{P}_e} \left( \boldsymbol{a}^\textrm{T}_t | \boldsymbol{s}^\textrm{T}_t \right)}{\pi_{\boldsymbol{W}^\textrm{P}_{e-1}} \left( \boldsymbol{a}^\textrm{T}_t | \boldsymbol{s}^\textrm{T}_t \right)} A_{\boldsymbol{W}^\textrm{V}} \left( \boldsymbol{s}^\textrm{T}_t, \boldsymbol{a}^\textrm{T}_t \right)$ encourages actions with positive advantage while penalizing those with negative advantage, thus improving the policy update. In particular, if action $\boldsymbol{a}^\textrm{T}_t$ results in a positive advantage (i.e., $A_{\boldsymbol{W}^\textrm{V}} \left( \boldsymbol{s}^\textrm{T}_t, \boldsymbol{a}^\textrm{T}_t \right) \geq 0$), and the robust-RL method increases the probability of selecting action $\boldsymbol{a}^\textrm{T}_t$ (i.e., $\pi_{\boldsymbol{W}^\textrm{P}_e} \left( \boldsymbol{a}^\textrm{T}_t | \boldsymbol{s}^\textrm{T}_t \right) > \pi_{\boldsymbol{W}^\textrm{P}_{e-1}} \left( \boldsymbol{a}^\textrm{T}_t | \boldsymbol{s}^\textrm{T}_t \right)$), we have $\frac{\pi_{\boldsymbol{W}^\textrm{P}_e} \left( \boldsymbol{a}^\textrm{T}_t | \boldsymbol{s}^\textrm{T}_t \right)}{\pi_{\boldsymbol{W}^\textrm{P}_{e-1}} \left( \boldsymbol{a}^\textrm{T}_t | \boldsymbol{s}^\textrm{T}_t \right)} A_{\boldsymbol{W}^\textrm{V}} \left( \boldsymbol{s}^\textrm{T}_t, \boldsymbol{a}^\textrm{T}_t \right) \geq 0$ and the loss is negative. If the loss is negative, the optimizer will try to make it more negative since the target of the optimizer is to minimize the loss function. Similarly, if action $\boldsymbol{a}^\textrm{T}_t$ results in a negative advantage (i.e., $A_{\boldsymbol{W}^\textrm{V}} \left( \boldsymbol{s}^\textrm{T}_t, \boldsymbol{a}^\textrm{T}_t \right) < 0$), and the robust-RL method increases the probability of selecting action $\boldsymbol{a}^\textrm{T}_t$, we have $\frac{\pi_{\boldsymbol{W}^\textrm{P}_e} \left( \boldsymbol{a}^\textrm{T}_t | \boldsymbol{s}^\textrm{T}_t \right)}{\pi_{\boldsymbol{W}^\textrm{P}_{e-1}} \left( \boldsymbol{a}^\textrm{T}_t | \boldsymbol{s}^\textrm{T}_t \right)} A_{\boldsymbol{W}^\textrm{V}} \left( \boldsymbol{s}^\textrm{T}_t, \boldsymbol{a}^\textrm{T}_t \right) < 0$ such that the loss is positive.
{\color{black} The adversarial loss is  
\begin{equation}\label{eq:loss_tilt_adv}
\begin{split}
    \mathcal{L}^\textrm{A} \left( \boldsymbol{W}^\textrm{P}_e, \rho_e \right) = \mathbb{E}_{b \sim \mathcal{D}^\textrm{T}_e} \! \left[ \! - \! \min \! \left( \! \frac{\hat{\pi}_{\boldsymbol{W}^\textrm{P}_e} \! \left( \boldsymbol{a}^\textrm{T}_t | \boldsymbol{s}^\textrm{T}_t; \rho_e \right)}{\pi_{\boldsymbol{W}^\textrm{P}_{e-1}} \! \left( \boldsymbol{a}^\textrm{T}_t | \boldsymbol{s}^\textrm{T}_t \right)} A_{\boldsymbol{W}^\textrm{V}} \! \left( \! \boldsymbol{s}^\textrm{T}_t, \boldsymbol{a}^\textrm{T}_t \!\right), \right. \right. \\ \left. \left. \textrm{clip} \! \left( \! \frac{\hat{\pi}_{\boldsymbol{W}^\textrm{P}_e} \! \left( \boldsymbol{a}^\textrm{T}_t | \boldsymbol{s}^\textrm{T}_t; \rho_e \right)}{\pi_{\boldsymbol{W}^\textrm{P}_{e-1}} \! \left( \boldsymbol{a}^\textrm{T}_t | \boldsymbol{s}^\textrm{T}_t \right)}, 1-\eta, 1+\eta \right) \! A_{\boldsymbol{W}^\textrm{V}} \! \left( \! \boldsymbol{s}^\textrm{T}_t, \boldsymbol{a}^\textrm{T}_t \! \right) \! \right) \right]. 
\end{split}
\end{equation}
Compared to the standard PPO loss function in (\ref{eq:loss_tilt_norm}), the adversarial loss function in (\ref{eq:loss_tilt_adv}) uses the worst-case policy defined in (\ref{eq:worst_case_policy}). Hence, the robust-RL must optimize its} {\color{black} performance under the worst-case policy thus improving its robustness. Given (\ref{eq:loss_tilt_norm}) and (\ref{eq:loss_tilt_adv}), the entire loss function used to train the robust-RL is 
\begin{equation}\label{eq:loss_tilt_policy}
    \mathcal{L} \left( \boldsymbol{W}^\textrm{P}_e, \rho_e \right) = \left( 1 - \kappa \right) \mathcal{L}^\textrm{N} \left( \boldsymbol{W}^\textrm{P}_e \right) + \kappa \mathcal{L}^\textrm{A} \left( \boldsymbol{W}^\textrm{P}_e, \rho_e \right), 
\end{equation}
where $\kappa$ a weight parameter that controls the trade-off between $\mathcal{L}^\textrm{N} \left( \boldsymbol{W}^\textrm{P}_e \right)$ and $\mathcal{L}^\textrm{A} \left( \boldsymbol{W}^\textrm{P}_e, \rho_e \right)$.} 
Given the loss function in (\ref{eq:loss_tilt_norm}), the update rule of the policy network $\boldsymbol{W}^\textrm{P}$ is 
\begin{equation}\label{eq:update_pi}
    \boldsymbol{W}^\textrm{P}_e \leftarrow \boldsymbol{W}^\textrm{P}_{e-1} + \alpha^\textrm{P} \nabla_{\boldsymbol{W}^\textrm{P}_{e-1}} \mathcal{L} \left( \boldsymbol{W}^\textrm{P}_{e-1}, \rho_{e-1} \right),  
\end{equation}
where $\alpha^\textrm{P}$ is the learning rate. {\color{black} By optimizing Eq. (\ref{eq:loss_tilt_policy}), the policy update is driven not only by the standard PPO loss but also by the adversarial term evaluated under the worst-case policy. In particular, in each training iteration, the agent is encouraged to improve the performance even under the worst-case scenario induced by the noise of the DNT data. As a result, the robust-RL can effectively mitigate the negative impact of the noise of the DNT data on model training, thus improving the learning stability of the robust-RL.} 
The loss function of updating the value network $\boldsymbol{W}^\textrm{V}$ is 
\begin{equation}\label{eq:loss_tilt_value}
    \mathcal{L}^\textrm{V} \left( \boldsymbol{W}^\textrm{V}_e \right) = \frac{1}{2} \mathbb{E}_{b \sim \mathcal{D}^\textrm{T}_e} \left[ A_{\boldsymbol{W}^\textrm{V}_e} \left( \boldsymbol{s}^\textrm{T}_t, \boldsymbol{a}^\textrm{T}_t \right) \right]^2. 
\end{equation}
Given the loss function, the update rule of the value network $\boldsymbol{W}^\textrm{V}$ is
\begin{equation}\label{eq:update_v}
    \boldsymbol{W}^\textrm{V}_e \leftarrow \boldsymbol{W}^\textrm{V}_{e-1} + \alpha^\textrm{V} \nabla_{\boldsymbol{W}^\textrm{V}_{e-1}} \mathcal{L}^\textrm{V} \left( \boldsymbol{W}^\textrm{V} \right),  
\end{equation}
where $\alpha^\textrm{V}$ is the learning rate. 

The training procedure of the first level RL is summarized in Algorithm \ref{alg:first_level}. 

\begin{algorithm}[!t]
    \small
    \caption{Training procedure of the first level robust-RL}
    \label{alg:first_level}
    \begin{algorithmic}
        \color{black}
        \FOR{each epoch $e$ of the first level robust-RL}
            \FOR{each time step $t$}
                \STATE The first level robust-RL obtains $\boldsymbol{s}^\textrm{T}_t$, with a subset $\mathcal{B}_e^\textrm{P}$ of data collected from the physical network, and a subset $\mathcal{B}_e^\textrm{D}$ of data collected from the DNT
                \STATE The first level robust-RL determines the action $\boldsymbol{a}^\textrm{T}_t$ (i.e., the antenna tilt angle $\boldsymbol{\psi}^\textrm{T}_t$) based on $\boldsymbol{s}^\textrm{T}_t$ and the current policy $\pi_{\boldsymbol{W}^\textrm{P}} \left( \boldsymbol{a}^\textrm{T}_t \mid \boldsymbol{s}^\textrm{T}_t \right)$
                \STATE The first level robust-RL calculates $R^\textrm{T} \left( \boldsymbol{s}_t^\textrm{T}, \boldsymbol{a}_t^\textrm{T} \right)$ based on (\ref{eq:reward_tilt}), and the transmission delay $\tau_{e,b} \left( \boldsymbol{\psi}^\textrm{T}_{t_{e,b}} \right)$ based on (\ref{eq: tau_e_b}). 
            \ENDFOR
            \STATE Record the collected transitions into the replay buffer
            \STATE Update the first level robust-RL based on (\ref{eq:loss_tilt_norm}) and~(\ref{eq:update_v})
        \ENDFOR
    \end{algorithmic}
\end{algorithm}

\subsection{Components of the Second Level PPO}
Next, we introduce the second level PPO for optimizing the data collection ratio \cite{schulman2017proximal}. {\color{black} Here, we use a standard PPO since the second level RL will not be trained by noisy data generated by DNT, and the impact of DNT data noise is handled by the first level robust-RL.} The second level PPO for data collection ratio optimization has the following seven components:  

\subsubsection{Agent} The agent is the BS. It observes the training performance of the robust-RL for tilt angle adjustment, and determines the data collection ratio for each episode $e$. 

\subsubsection{State} The state of the BS describes the robust-RL training performance
in epoch $e-1$. We use the expected sum of reward $\frac{1}{ \left| \mathcal{B}_{e-1} \right| } \sum_{b \in \mathcal{B}_{e-1}} R^\textrm{T} \left( \boldsymbol{s}^\textrm{T}_{t_{e-1,b}}, \boldsymbol{a}^\textrm{T}_{t_{e-1,b}} \right)$ of an episode and the policy network loss $\mathcal{L} \left( \boldsymbol{W}^\textrm{P}_{e-1}, \rho_{e-1} \right)$ to represent the robust-RL training status. 
Thus, the status of the BS at each step $e$ of the second level PPO is $\boldsymbol{s}_e = \left[ \mathcal{L} \left( \boldsymbol{W}^\textrm{P}_{e-1}, \rho_{e-1} \right), \frac{1}{ \left| \mathcal{B}_{e-1} \right| } \sum_{b \in \mathcal{B}_{e-1}} R^\textrm{T} \left( \boldsymbol{s}^\textrm{T}_{t_{e-1,b}}, \boldsymbol{a}^\textrm{T}_{t_{e-1,b}} \right) \right]$. \cite{parker2022automated}

\subsubsection{Action} The action determines the ratio of data collected from the physical network and the DNT. Thus, each action at training epoch $e$ is $a_e = \rho_e$. 

\subsubsection{Reward function} The reward function $R \left( \boldsymbol{s}_e, a_e \right)$ evaluates action $a_e$ based on state $\boldsymbol{s}_e$. $R \left( \boldsymbol{s}_e, a_e \right)$ is given by 
\begin{equation}\label{eq:reward_ratio}
    \begin{split}
        R & \left( \boldsymbol{s}_e, a_e \right) = \frac{1}{ \left| \mathcal{B}_e \right| } \sum_{b \in \mathcal{B}_e} R^\textrm{T} \left( \boldsymbol{s}^\textrm{T}_{t_{e,b}}, \boldsymbol{a}^\textrm{T}_{t_{e,b}} \right) \\
        &- \mathbbm{1}_{ \{ \sum_{b=1}^{\rho_e \left| \mathcal{B}_e \right|} \! \tau_{e,b} \left( \! \boldsymbol{\psi}^\textrm{T}_{t_{e,b}} \! \right) > \tau_{max} \}} \xi \left( \! \sum_{b=1}^{\rho_e \left| \mathcal{B}_e \right|} \! \tau_{e,b} \! \left( \! \boldsymbol{\psi}^\textrm{T}_{t_{e,b}} \! \right) \! - \! \tau_{max} \! \right), 
    \end{split} 
\end{equation}
where $\xi > 0$ is the penalty parameter for the data collection delay in training epoch $e$ of the first level robust-RL exceeding the maximum threshold $\tau_{max}$. The definitions of the policy function, value function, and advantage function are similar to those of the first level robust-RL. 

The training process of the second level PPO is similar to that of the first level robust-RL. In particular, the loss function for updating the second level PPO is the clipped surrogate function in (\ref{eq:loss_tilt_norm}). 
To train the second level RL, we still need to collect a set of transitions. Here, a transition of the second level PPO is collected when the first level robust-RL completes one training epoch. Based on the loss function and the collected transitions, the second level PPO is updated via a stochastic gradient descent method, which is similar to (\ref{eq:update_pi}). The entire training procedure of our proposed hierarchical RL is summarized in Algorithm \ref{alg:second_level}.

\begin{algorithm}[!t]
    \small
    \caption{Training procedure of the proposed hierarchical RL}
    \label{alg:second_level}
    \begin{algorithmic}
        \color{black}
        \FOR{each training epoch of the second level PPO}
            \STATE The second level PPO obtains $\boldsymbol{s}_e$ from the training procedure of the first level robust-RL
            \STATE The second level PPO determines the data collection ratio $\rho_e$ for the first level robust-RL (i.e., $a_e$) based on $\boldsymbol{s}_e$ and the current policy
            \STATE Implement one training epoch $e$ of the first level robust-RL with ratio $\rho_e$ according to Algorithm \ref{alg:first_level} and obtain the transmission delay $\tau_{e,b} \left( \boldsymbol{\psi}^\textrm{T}_{t_{e,b}} \right)$
            \STATE The second level PPO calculates $R \left( \boldsymbol{s}_e, a_e \right)$ based on (\ref{eq:reward_ratio})
        \ENDFOR
        \STATE Record the collected transitions into a reply buffer
        \STATE Update the second level PPO 
    \end{algorithmic}
\end{algorithm}

\section{Convergence, Complexity and Implementation Analysis of the Proposed Algorithm}
In this section, we analyze the convergence, the complexity and the implementation of the proposed hierarchical RL framework. 

\subsection{Convergence Analysis}\label{se:analysis}
To analyze the convergence of our proposed method. We first introduce the following conditions. 


\begin{condition}\label{asp1}
    \emph{The reward function $R \left( \boldsymbol{s}_e, a_e \right)$ of the PPO is bounded. }
\end{condition}

\begin{condition}\label{asp2}
    \emph{
    A constant $L$ exists such that the policy $\pi_{\overline{\boldsymbol{W}}} \left( \boldsymbol{s}_e, a_e\right)$ of the second level PPO is L-smooth with respect to network parameters $\overline{\boldsymbol{W}}$, where $\overline{\boldsymbol{W}}$ represents the parameters of both the policy and value networks. Hence, for any network parameters $\overline{\boldsymbol{W}}_1$, $\overline{\boldsymbol{W}}_2$, we have }
    \begin{equation}
            \lVert \nabla \pi_{\overline{\boldsymbol{W}}_1} \left( \boldsymbol{s}_e, a_e\right) - \nabla \pi_{\overline{\boldsymbol{W}}_2} \left( \boldsymbol{s}_e, a_e\right) \rVert \leq L \lVert \overline{\boldsymbol{W}}_1 - \overline{\boldsymbol{W}}_2 \rVert. 
    \end{equation}
\end{condition}

\begin{condition}\label{asp3}
    \emph{The learning rate $\alpha_k$ used in epoch $k$ for updating the PPO satisfies $\sum_{k=1}^{+ \infty} \alpha_k = + \infty$ and $\alpha_k \to 0$. }
\end{condition} 

\begin{condition}\label{asp4}
    \emph{The PPO advantage function approximated by $\overline{\boldsymbol{W}}$ is 
    bounded 
    by $\varphi_k$ at each 
    epoch $k$. Furthermore, the convergence of the first level robust-RL policy leads to a progressive reduction in $\varphi_k$. }
\end{condition}

Given the above conditions, the convergence of the second level PPO is analyzed in the following corollary. 

\begin{corollary}\label{th1}
    \emph{If the second level PPO meets conditions \ref{asp1} - \ref{asp4}, the gradient norm $\| \nabla V_{\overline{\boldsymbol{W}}_k} \|$ at epoch $k$ satisfies: }
    \begin{equation}
        \begin{split}
            \min_{k} \mathbb{E} \left[ \lVert \nabla V_{\overline{\boldsymbol{W}}_k} \rVert^2 \right] = O \left( \frac{\sum_k \alpha_k^2}{\sum_k \alpha_k} \right) + O \left( {\lim \sup}_{k \to + \infty} \varphi_k \right), 
        \end{split}
    \end{equation}
    \emph{and}
    \begin{equation}
        \begin{split}
            \frac{1}{\sum_k \alpha_k} &\sum_k \alpha_k \mathbb{E} \left[ \lVert \nabla V_{\overline{\boldsymbol{W}}} \rVert^2 \right] = \\
            &O \left( \frac{\sum_k \alpha_k^2}{\sum_k \alpha_k} \right) + O \left( {\lim \sup}_{k \to + \infty} \varphi_k \right), 
        \end{split}
    \end{equation}
\emph{where $O \left( \cdot \right)$ represents the asymptotic upper bound. Hence, the second level PPO converges to a stationary point.} 
\end{corollary}

\begin{proof}
    To prove Corollary \ref{th1}, we need to show that the defined second level PPO meets conditions \ref{asp1} - \ref{asp4}. With regard to condition \ref{asp1}, 
    since the downlink data rate $r_{uc,t}^\textrm{D} \left( \boldsymbol{\psi}_t^\textrm{T} \right)$ of each user is bounded, the reward $R^\textrm{T} \left( \boldsymbol{s}_t^\textrm{T}, \boldsymbol{a}_t^\textrm{T} \right)$ of the first level PPO is bounded. 
    From (\ref{eq:reward_ratio}), we see that the reward of the second level PPO is the average reward of the first level robust-RL over epoch $e$ (i.e., $\left| \mathcal{B}_e \right|$ transitions). Hence, the reward of the second level PPO is also bounded such that our designed
    second level PPO satisfies Condition \ref{asp1}. 
    
    To prove that the second level PPO meets Condition \ref{asp2}, we just to prove that the multi layer perceptron (MLP) model is Lipschitz continuous since we use a MLP to approximate the value function of the second level PPO \cite{NEURIPS2018_d54e99a6, NEURIPS2022_0ff54b4e}. Let $\overline{\boldsymbol{w}}_l$ be the weight matrix and $b_l$ be the bias of each layer $l$, followed by an activation function layer $\sigma \left( \cdot \right)$. For each fully connected layer $l$ without the activation function, we can always find a constant $L$ such that the following inequality holds for any pair of inputs $x_1, x_2$ to that layer  
    \begin{equation}
        \| \overline{\boldsymbol{w}}_l x_1 + b_l - \left( \overline{\boldsymbol{w}}_l x_2 + b_l \right) \| \leq L \| x_1 - x_2 \|.
    \end{equation}
    If we consider the activation function, we provide the following proof. Let $y_1$ and $y_2$ be the input of the activation function layer which are close enough. Hence, we can use the first order Taylor expansion to expand $\sigma \left( y_1 \right)$: 
    \begin{equation}
        \sigma \left( y_1 \right) = \sigma \left( y_2 \right) + \frac{ \partial \sigma \left( y_2 \right)}{\partial y_2} \left( y_1 - y_2 \right). 
    \end{equation}
    If we move $\sigma \left( y_2 \right)$ to the left side of the equation, we have 
    \begin{equation}
        \sigma \left( y_1 \right) - \sigma \left( y_2 \right) = \frac{ \partial \sigma \left( y_2 \right)}{\partial y_2} \left( y_1 - y_2 \right). 
    \end{equation}
    For common activation functions such as sigmoid, tanh, and ReLU, their derivatives are uniformly bounded. Therefore, the term $ \frac{\partial \sigma \left( y_2 \right)}{\partial y_2}$ can be treated as a constant $Y$. As a result, we can still find a constant $L$ such that
    \begin{equation}
        \begin{split}
            \| \sigma \left( y_1 \right) - \sigma \left( y_2 \right) \| &= \| Y \left( y_1 - y_2 \right) \| \leq L \| y_1 - y_2 \|. 
        \end{split}
    \end{equation}
    Hence, the second level PPO satisfies Condition \ref{asp2}. 

    To meet Condition \ref{asp3}, the second level PPO can use a diminishing learning rate such that the learning rate $\alpha_k$ satisfies $\sum_{k=1}^{+ \infty} \alpha_k = + \infty$ and $\alpha_k \to 0$. This learning rate setting method ensures that the PPO continues to make progress while gradually reducing the step size, which is a standard requirement for convergence in stochastic approximation frameworks.  

    With regard to Condition \ref{asp4}, since we use neural networks to approximate the value function and compute the advantage function in the second level PPO, there exists an estimate error $\varphi_k > 0$ between the estimated and actual advantage function. Since the states, actions, and rewards in the second level PPO depend on the training process of the first level robust-RL, as the robust-RL gradually converges, the estimation error $\varphi_k$ of the second level PPO value function also decreases progressively and eventually approaches a minimum value. Hence, the second level PPO meets Condition \ref{asp4}. 

    When second level PPO meets Conditions \ref{asp1} - \ref{asp4}, it converges to a stationary point based on Theorem 3.3 from \cite{jin2023stationary}. This completes the proof. 
\end{proof}
{\color{black} From Corollary \ref{th1}, we see that as $k$ increases, the terms $O \left( \frac{\sum_k \alpha_k^2}{\sum_k \alpha_k} \right)$ and $O \left( {\lim \sup}_{k \to + \infty} \varphi_k \right)$ decrease. Hence, $\min_{k} \mathbb{E} \left[ \lVert \nabla V_{\overline{\boldsymbol{W}}_k} \rVert^2 \right]$ and $\frac{1}{\sum_k \alpha_k} \sum_k \alpha_k \mathbb{E} \left[ \lVert \nabla V_{\overline{\boldsymbol{W}}} \rVert^2 \right]$ decrease. When the gradient becomes sufficiently small, the updates to the second level PPO model are negligible. As a result, the second level PPO converges. }

\subsection{Complexity Analysis}
Next, we analyze the complexity of our proposed method. Both the first level robust-RL and the second level PPO employ a structure consisting of a policy network and a value network, where both networks are implemented as fully connected neural networks. We assume that the first level robust-RL policy network contains $L^{\textrm{P}_1}$ layers, with $W^{\textrm{P}_1}_l$ neurons in layer $l$, and the value network contains $L^{\textrm{V}_1}$ layers, with $W^{\textrm{V}_1}_l$ in layer $l$. Following \cite{10750404, 11270936}, we consider the number of scalar multiplications per update as the main computational cost. Thus, the computational complexity of the first level robust-RL training process is 
\begin{equation}\label{eq:cplx_1}
    \mathcal{O} \left( \sum_{e=1}^{E_1} \left( \left| \mathcal{B}_e \right| \left(\sum_{l=1}^{L^{\textrm{P}_1}} W^{\textrm{P}_1}_{l-1} W^{\textrm{P}_1}_{l} + \sum_{l=1}^{L^{\textrm{V}_1}} W^{\textrm{V}_1}_{l-1} W^{\textrm{V}_1}_{l} \right) \right) \right), 
\end{equation}
where $E_1$ is the number of epochs of training the first level robust-RL. Given (\ref{eq:cplx_1}), the computational complexity of the second level PPO depends on that of the first level robust-RL, since the second level PPO determines the data collection behavior used in each epoch of the first level robust-RL training process. Similar to the first level robust-RL, we assume that the second level PPO includes a policy network with $L^{\textrm{P}_2}$ layers and $W^{\textrm{P}_2}_l$ neurons in each layer $l$, and a value network with $L^{\textrm{V}_2}$ layers and $W^{\textrm{V}_2}_l$ neurons in each layer $l$. Then, the computational complexity of training the second level PPO is 
\begin{equation}\label{eq:cplx_2}
    \begin{split}
        \mathcal{O} \left( \sum_{e=1}^E \left( \left| \mathcal{B}_e \right| \left(\sum_{l=1}^{L^{\textrm{P}_1}} W^{\textrm{P}_1}_{l-1} W^{\textrm{P}_1}_{l} + \sum_{l=1}^{L^{\textrm{V}_1}} W^{\textrm{V}_1}_{l-1} W^{\textrm{V}_1}_{l} \right) \right) \right. \\ 
        \left. + E \left(\sum_{l=1}^{L^{\textrm{P}_2}} W^{\textrm{P}_2}_{l-1} W^{\textrm{P}_2}_{l} + \sum_{l=1}^{L^{\textrm{V}_2}} W^{\textrm{V}_2}_{l-1} W^{\textrm{V}_2}_{l} \right) \right). 
    \end{split}
\end{equation}
{\color{black} In (\ref{eq:cplx_2}), the first term represents the complexity of the first level robust-RL training, and the second term represents the complexity of the second level PPO collecting its training data from the first level robust-RL training process. If the two level networks have comparable sizes, the dominant term is the first term, $ \! \sum_{e=1}^E \! |\mathcal{B}_e| \! \left( \! \sum_{l=1}^{L^{\textrm{P}_1}} \! W^{\textrm{P}_1}_{l-1} \! W^{\textrm{P}_1}_{l} \! + \! \sum_{l=1}^{L^{\textrm{V}_1}} \!  W^{\textrm{V}_1}_{l-1} \! W^{\textrm{V}_1}_{l} \! \right)$, since it scales with the batch size $|\mathcal{B}_e|$ in addition to the network sizes. In contrast, the second term, $E\left(\sum_{l=1}^{L^{\textrm{P}_2}} W^{\textrm{P}_2}_{l-1}W^{\textrm{P}_2}_{l}+\sum_{l=1}^{L^{\textrm{V}_2}} W^{\textrm{V}_2}_{l-1}W^{\textrm{V}_2}_{l}\right)$, only scales linearly with $E$ and the second level network sizes. Therefore, under the considered settings where $|\mathcal{B}_e|\gg 1$, the first term dominates the overall complexity.} 

{\color{black} \subsection{Implementation of the Proposed Method}
Next, we analyze the implementation of our proposed method, which includes: 1) determining the antenna tilt angles by the first level robust-RL, and 2) selecting the ratio $\rho_e$ of data to be collected from the physical network and the DNT by the second level PPO during each robust-RL training epoch. To implement the first level robust-RL, first, the BS must be aware of several parameters of the physical network, including the number $U$ of users, the number $C$ of cells, the location of the BS, the transmit power $P_0$, the antenna gain $g^\textrm{U}_{ut}$ of each user $u$, the pass loss exponent $\beta$, the vertical and horizontal beamwidths $\theta^\textrm{V}$ and $\theta^\textrm{H}$, and the bandwidth $B$. The BS also needs to collect training data in form of tuples $\left( \boldsymbol{s}^\textrm{T}_t, \boldsymbol{a}^\textrm{T}_t, R^\textrm{T} \left( \boldsymbol{s}^\textrm{T}_t, \boldsymbol{a}^\textrm{T}_t\right) \right)$. These data can be collected either from the physical network or the DNT. 
The number $\left| \mathcal{B}_e \right|$ of tuples that the BS must collect to train the first level robust-RL is determined by the second level PPO. To implement the second level PPO, the BS must observe training feedback from the first level robust-RL, such as the ratio $\rho_e$, the average episode rewards, and the policy network updates. It also needs to collect training data in form of tuples $\left( \boldsymbol{s}_e, a_e, R \left( \boldsymbol{s}_e, a_e \right) \right)$. Based on this information, the PPO dynamically adjusts the data collection ratio $\rho_e$ to balance the robustness and efficiency of the first level robust-RL training. }

\section{Simulation Results}\label{se:simulation}
For simulations, we consider a wireless network where the BS is located at $\left[ 0, 0 \right]$ and the coverage radius of the network is 50 meters. The BS is equipped with three directional antennas. Each antenna covers a sector shaped cell and serves the users located within that cell. $U=10$ users are randomly distributed within the BS coverage. We assume that the DNT data generation error follows a uniform distribution over a bounded interval \cite{pinto2017asymmetric}. For example, $\epsilon \left( \boldsymbol{l}_{u,t}^\textrm{U} \right)\in \left[ -\varepsilon, \varepsilon \right]$, with $\varepsilon$ being the bound of the DNT data generation error. Other parameters used in the simulations are listed in Table \ref{tab:sys_params}. {\color{black} For comparison purposes, we consider two baselines: 
\begin{itemize}
    \item In the first baseline, antenna tilt angles are adjusted by the robust-RL, as done in our proposed method. The data collection ratio of training the first level robust-RL is randomly determined. This baseline is used to demonstrate the impact of the second level PPO on the performance of the first level PPO. 
    \item In the second baseline, antenna tilt angles are adjusted by a vanilla PPO, and the data collection ratio of training the first level PPO is determined by another vanilla PPO. 
    This baseline is used to evaluate the performance gain achieved by the proposed first level robust-RL. 
\end{itemize}}

\begin{table}[!t]
    \vspace{-0.2cm}
    \caption{System Parameters {\color{black}\cite{zhang2023optimization, 8638796, NEURIPS2021_dbb42293}}}
    \label{tab:sys_params}
    \centering
    \begin{tabular}{|c|c|c|c|c|c|}
        \hline
        \textbf{Param} & \textbf{Value} & \textbf{Param} & \textbf{Value} &
        \textbf{Param} & \textbf{Value} \\
        \hline
        $U$ & 10 & $C$ & 3 & $P_0$ & 1 \\
        \hline
        $g^\textrm{U}_{ut}$ & 1 & $\delta_{uct}$ & 1 & $a$ & 1 \\
        \hline 
        $\beta$ & 2 & $N_0$ & $1 \times 10^{-6}$ & $W$ & 1 \\
        \hline
        $\theta^{\textrm{V}}$ & $30^\circ$ & $\theta^{\textrm{H}}$ & $120^\circ$ & $N$ & 3 \\
        \hline
        $\phi^\textrm{A}_{1t}$ & $0^\circ$ & $\phi^\textrm{A}_{2t}$ & $120^\circ$ & $\phi^\textrm{A}_{3t}$ & $240^\circ$ \\
        \hline
        $D$ & 1 & $\tau_{max}$ & 150 & $\psi^\textrm{T}_{min}$ & $0^\circ$ \\
        \hline
        $\psi^\textrm{T}_{max}$ & $90^\circ$ & $\lambda^\textrm{T}$ & 0.99 & $\alpha^\pi$ & $3 \times 10^{-3}$ \\
        \hline
        $\alpha^\textrm{V}$ & $3 \times 10^{-4}$ & $\xi$ & 0.005 & & \\
        \hline
    \end{tabular}
    \vspace{-0.3cm}
\end{table}  
Fig. \ref{fig:overhead} shows the delay of the BS collecting data from the physical network during the first level RL training process. From Fig. \ref{fig:overhead}, we see that both our proposed method and the second baseline can effectively reduce the time delay. This is because the second level PPO in both proposed method gradually learns an efficient data collection policy that efficiently uses the DNT data for first level RL training. We can also see that our proposed method reduces the data collection delay by up to 28.01\%, compared to the second baseline. This is because the first level robust-RL in our proposed method effectively mitigates the adverse effect of the noisy training data, thus can exploit more noisy samples collected from the DNT without degrading the training performance. 

\begin{figure}[!t]
  \begin{center}
    \includegraphics[width=9cm]{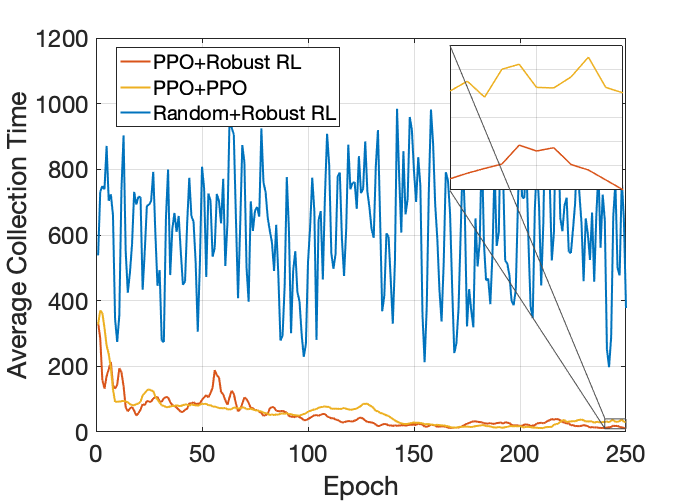}
    \vspace{-0.3cm}
    \caption{\small{The delay of collecting data from the physical network as the number epochs varies. }} 
    \label{fig:overhead}
    \vspace{-0.7cm}
  \end{center}
\end{figure}


In Fig. \ref{fig:converge}, we show the convergence performance of the second level PPO in the proposed methods and two baselines. From Fig. \ref{fig:converge}, we see that, as the number of training epochs increases, the average episode return of the second level PPO in both methods increases. This implies that the second level RL effectively learns a policy that balances the data collection from the physical network and the DNT. We also see that the second level PPO in our proposed method outperforms the second baseline with a 77.81\% higher average episode return. This is due to the fact that the robust-RL in our proposed method resists noise in the training data and achieves better training performance with higher episode rewards compared with the standard PPO. 

\begin{figure}[!t]
  \begin{center}
    \includegraphics[width=9.5cm]{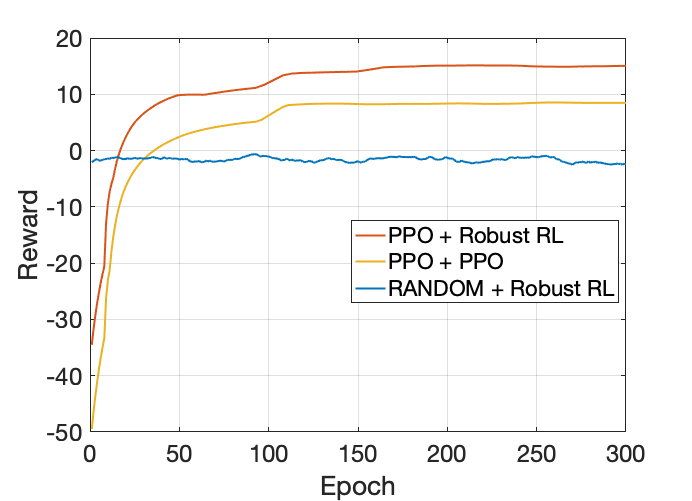}
    \vspace{-0.3cm}
    \caption{\small{The convergence of the second level PPO. }} 
    \label{fig:converge}
  \end{center}
\end{figure}

In Fig. \ref{fig:converge_1}, we show the convergence performance of the first level RL in the proposed methods and the second baseline. From Fig. \ref{fig:converge_1}, we can see that the average episode return of both the robust-RL and the vanilla PPO increase as the number of training epochs increases. This is because both methods can learn the relationship between antenna tilt angles and user dynamics. We can also observe that the robust-RL, which we use as the first level RL in our proposed method, improves the average episode reward by 38.51\% compared with the vanilla PPO. This improvement can be attributed to the fact that the robust-RL exhibits higher robustness to noise in the training data, thus can achieve better training performance compared with the vanilla PPO under noisy scenarios. 

\begin{figure}[!t]
  \begin{center}
  \color{black}
    \includegraphics[width=9cm]{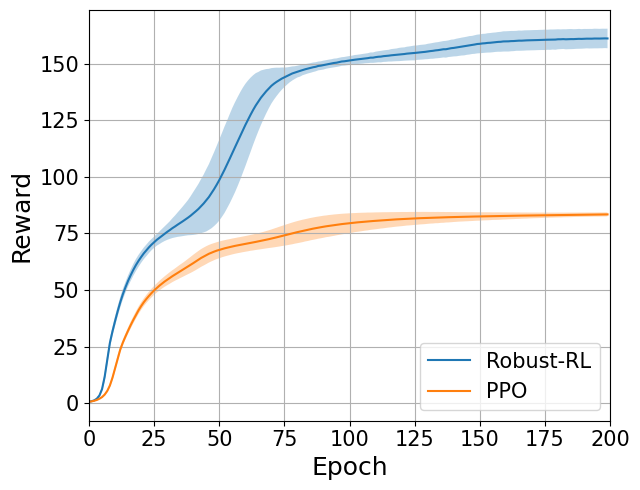}
    \vspace{-0.3cm}
    \caption{\small{The convergence of the first level RL. }} 
    \label{fig:converge_1}
    \vspace{-0.7cm}
  \end{center}
\end{figure}


Fig. \ref{fig:epsilon} shows the performance of the second level PPO under different error level $\varepsilon$. From Fig. \ref{fig:epsilon} we see that the second level PPO can converge under different error levels. This is because the first level robust-RL incorporates a worst-case policy in the adversarial loss term of the loss function. By optimizing the worst case performance induced by the DNT data errors, the antenna adjustment policy is robust to those errors. Then, the second level PPO can receive more reliable data from the first level robust-RL training procedure. Fig. \ref{fig:epsilon} also shows that the second level PPO with $\varepsilon = 0.05$ achieves a higher average episode reward compared to the second level PPO with $\varepsilon = 0.25$. This is because when the error level decreases from 0.25 to 0.05, the first level robust-RL can learn a more accurate antenna adjustment policy, which in turn provides higher average episode rewards to the second level PPO.  

\begin{figure}[!t]
  \begin{center}
    \includegraphics[width=9.5cm]{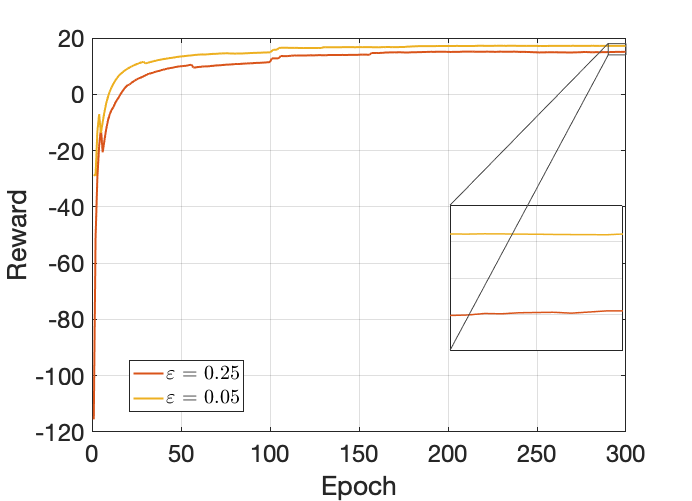}
    \vspace{-0.3cm}
    \caption{\small{Performance of the second level PPO under different DNT error levels $\varepsilon$. }} 
    \label{fig:epsilon}
    \vspace{-0.7cm}
  \end{center}
\end{figure}


In Fig. \ref{fig:user_num}, we show how the episode reward achieved by the second level RL as the number of users varies. From Fig.~\ref{fig:user_num}, we see that the episode reward at the convergence of the proposed method and the second baseline increase as the number of users increases. This is due to the fact that the second level RL can effectively adjust the data collection ratio and maintain a balance between the learning efficiency of the first level RL and the time delay introduced by data collection. Fig. \ref{fig:user_num} also shows that our proposed method outperforms the second baseline (PPO + PPO) by up to 73.99\%. This is because the first level robust-RL incorporates an adversarial loss in its loss function to capture the impacts of DNT data errors on the antenna adjustment policy training. Hence, the designed RL method not only optimizes the antenna adjustment policy, but also optimizes the policy robustness to DNT data errors. 

\begin{figure}[!t]
  \begin{center}
    \includegraphics[width=9.5cm]{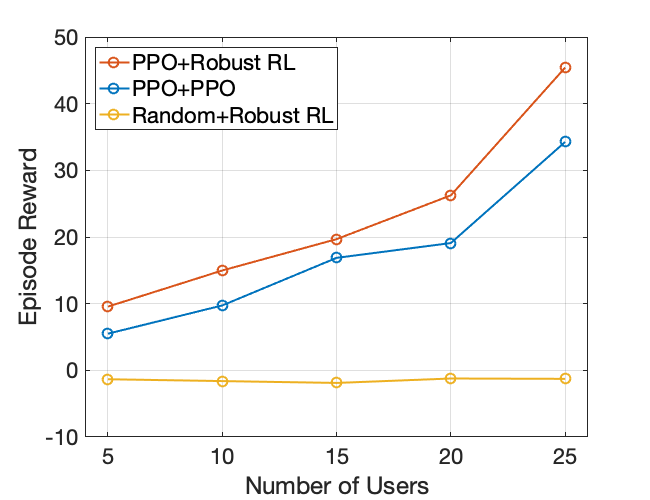}
    \vspace{-0.3cm}
    \caption{\small{The optimal performance of the second level RL as the number of users varies. }} 
    \label{fig:user_num}
    \vspace{-0.7cm}
  \end{center}
\end{figure}

In Fig. \ref{fig:k}, we show the convergence of the first level robust-RL as the value $\kappa$ of the weight that balances the importance between the standard PPO loss and the adversarial loss. From Fig. \ref{fig:k}, we see that the average episode reward of the first level robust-RL increases as $\kappa$ increases. 
This is because the loss function focuses more on the adversarial loss term, thus improving the learning performance of the first level robust-RL under noisy DNT data. 

\begin{figure}[!t]
  \begin{center}
    \includegraphics[width=9.5cm]{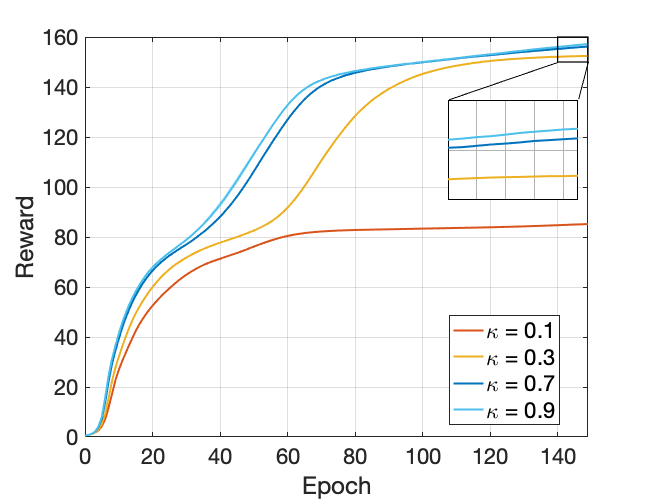}
    \vspace{-0.3cm}
    \caption{\small{The performance of first level robust-RL as value of $\kappa$ varies. }} 
    \label{fig:k}
    \vspace{-0.7cm}
  \end{center}
\end{figure}

In Fig. \ref{fig:xi}, we show the convergence of the second level PPO under different values of the weight $\xi$ of the penalty on the data collection time delay in the second level PPO reward function. 
In Fig. \ref{fig:xi}, we observe that when $\xi$ increases, the convergence speed of the designed method decreases. This is because when $\xi$ increases, the penalty introduced by the excessive delay increases. 
Hence, the gradient updates are primarily driven by the penalty and aim to meet the constraint (\ref{eq:const2}) instead of maximizing the objective function in (\ref{eq:problem}). 
This scenario is more pronounced in the early training stage. For example, in Fig. \ref{fig:xi}, during the first 100 epochs, the second level PPO with $\xi = 0.1$ has the largest penalty among all tested $\xi$ values, and hence, its convergence speed is the lowest. 

\begin{figure}[!t]
  \begin{center}
    \includegraphics[width=9.5cm]{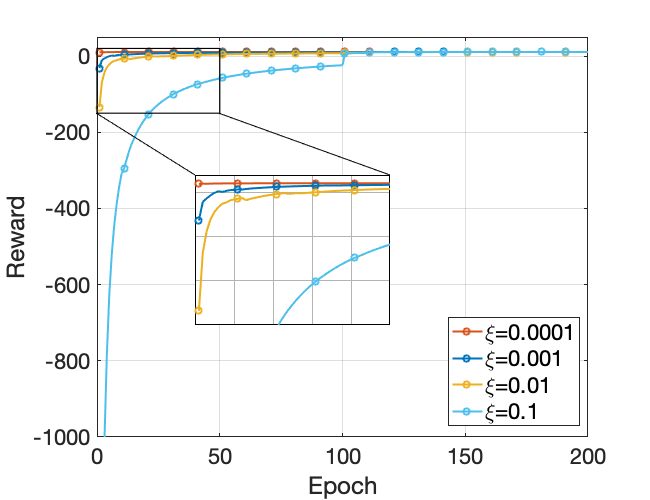}
    \vspace{-0.3cm}
    \caption{\small{The optimal performance of the second level RLs as the number of users varies. }} 
    \label{fig:xi}
    \vspace{-0.7cm}
  \end{center}
\end{figure}
 
\section{Conclusion}\label{se:conclusion}
In this paper, we have proposed a novel DNT assisted wireless DL model training framework that enables central controllers to dynamically select the training data from both the physical network and the DNT to optimize the training
process of a DL model used for physical network performance optimization according to network dynamics and DL training settings. 
In particular, since the data collected from the physical network is more accurate but incurs higher communication overhead compared to the data collected from the DNT, the BS must determine the data collection ratio to balance RL training performance and communication delay caused by physical network data collection. We have formulated the joint tilt angle adjustment and the data collection ratio selection as an optimization problem, aiming to maximize the data rates of all users, while considering the physical network data collection delay. To solve this optimization problem, we have proposed a hierarchical RL framework that integrates the robust adversarial loss-RL and the PPO. In our proposed framework, the first level robust-RL optimizes the antenna tilt angles, and the second level PPO dynamically adjusts the data collection ratio to improve the training performance of the first level robust-RL. 
Simulation results have shown that our proposed method significantly reduces the physical network data collection delay compared to the baselines. 


\bibliographystyle{IEEEbib}
\bibliography{references1}
\end{document}